\documentclass[a4paper, 11pt]{amsart}
\usepackage{a4wide}
\usepackage{latexsym, fullpage}
\usepackage[left=2cm,top=3cm,right=2cm]{geometry}
\usepackage[utf8]{inputenc}
\usepackage{amsmath}
\usepackage{amsfonts}
\usepackage{amsthm}
\usepackage{amssymb}
\usepackage{amscd}
\usepackage[all]{xy}
\usepackage{fancyhdr}
\usepackage{enumerate}
\usepackage{accents}
\usepackage{setspace}
\usepackage{xcolor}
\usepackage{upgreek}
\usepackage{verbatim}
\usepackage[mathscr]{eucal}
\usepackage{mathrsfs}
\usepackage[numbers]{natbib}
\usepackage[fit]{truncate}
\usepackage{graphicx}
\usepackage{epsfig}
\usepackage{soul}

\usepackage{tabularx,ragged2e,booktabs,caption}
\newcolumntype{C}[1]{>{\Centering}m{#1}}

\numberwithin{equation}{section}

\theoremstyle{plain}
\newtheorem{theorem}{Theorem}[section]
\newtheorem{corollary}[theorem]{Corollary}
\newtheorem{lemma}[theorem]{Lemma}
\newtheorem{proposition}[theorem]{Proposition}

\newtheorem{remark}[theorem]{Remark}
\newtheorem{assumption}{Assumption}

\theoremstyle{definition}
\newtheorem{definition}[theorem]{Definition}

\theoremstyle{remark}

\newcommand{\R}{\mathbf{R}}

\newcommand{\N}{\mathbf{N}}

\newcommand{\Lip}{\operatorname{Lip}}

\newcommand{\argmax}{\operatornamewithlimits{argmax}}
\newcommand{\norm}[1]{\left\lVert#1\right\rVert}

\makeatletter
\def\munderbar#1{\underline{\sbox\tw@{$#1$}\dp\tw@\z@\box\tw@}}
\makeatother

\allowdisplaybreaks

\title{Existence in Multidimensional Screening with General Nonlinear Preferences}

\thanks{
	$^*$ This paper is based on Chapter 3 of the author's thesis~\cite{Zhang18}.
	The author would like to express his deepest gratitude to his Ph.D. advisor Robert J. McCann for leading him to this project, and for his guidance and inspiration throughout. The author is grateful to Xianwen Shi, Alfred Galichon, Guillaume Carlier and  Ivar Ekeland  for stimulating conversations and encouragement,
	as well as to Georg N\" oldeke and Larry Samuelson for sharing their work in preprint form and vital remarks. 
	This project was initiated during the Fall of 2013 when the author was in residence at the Mathematical Sciences Research Institute in Berkeley CA, under a program supported by National Science Foundation Grant No. 0932078 000, and progressed during the Fall 2014 program of the Fields Institute for the Mathematical Sciences.
	\copyright \today}

\author{Kelvin Shuangjian Zhang$^\dagger$}\thanks{$^\dagger$Department of Mathematics, University of Toronto, Toronto, Ontario, Canada, M5S 2E4 {\tt szhang@math.toronto.edu}}

\begin{document}

\begin{abstract}
	We generalize the approach of Carlier (2001) and provide an existence proof for the multidimensional screening problem with general nonlinear preferences. We first formulate the principal's problem
	as a maximization problem with $G$-convexity constraints and then use $G$-convex analysis to prove existence. \medskip 

\end{abstract}

\bigskip

\maketitle

\section{Introduction}\label{section:introduction}
	This paper provides a general existence for a multidimensional nonlinear pricing model,  which is a natural extension of the models studied by Mussa-Rosen \cite{MussaRosen78}, Spence \cite{Spence74, Spence80}, Myerson \cite{Myerson81}, Baron-Myerson \cite{BaronMyerson82}, Maskin-Riley \cite{MaskinRiley84}, Wilson \cite{Wilson93}, Rochet-Chon\'e \cite{RochetChone98}, Monteiro-Page \cite{MonteiroPage98} and  Carlier~\cite{Carlier01}. A significant distinction lies in whether the agents' private type is one-dimensional (such as \cite{MussaRosen78, MaskinRiley84}), or multidimensional (such as \cite{RochetChone98,MonteiroPage98, Carlier01}). Another distinction is whether preferences are quasilinear on price (such as \cite{Armstrong96, Carlier01}) or fully nonlinear (such as \cite{NoldekeSamuelson15p, McCannZhang17}), especially for multidimensional models.  
	\medskip
	
	This paper proves the existence of a (price menu) solution to a multidimensional multiproduct monopolist problem, by extending Carlier \cite{Carlier01} to fully nonlinear preferences. $G$-convex analysis,  which is strongly tied to Trudinger's theory on the regularity of nonlinear PDEs \cite{Trudinger14}  developed for vastly different purposes, is employed to deal with the difficulty of non-quasilinear preferences. This method is potentially applicable to other problems under the same principal-agent framework, such as the study of tax policy (\cite{Mirrlees71}) and other regulatory policies (\cite{BaronMyerson82}). \medskip

	Consider the problem for a multiproduct monopolist (the principal) who sells indivisible products to a population of consumers (agents), who each buy at most one unit. Assume the monopolist is able to produce enough of each product such that there are neither product supply shortages {nor economies of scale}. Taking into account participation constraints and incentive compatibility, the monopolist would like to find the optimal menu of prices to maximize her total profit.\medskip

	In this paper, we first identify incentive compatibility with a $G$-convexity constraint, before rewriting the maximization problem by converting the optimization variables from a product-price pair of mappings to a product-value pair. It can then be shown that the product-value pair converges under the $G$-convexity constraint. The existence result follows. \medskip

	Starting from Mirrlees \cite{Mirrlees71} and Spence \cite{Spence74}, there are two main types of generalizations. One generalization is regarding dimension, from one-dimensional to multi-dimensional. The other generalization is in the form of utility functions, to beyond quasilinear.\medskip

	For the quasilinear case, where the utility function depends linearly on price, theories of existence \cite{Basov05,RochetStole03,Carlier01,MonteiroPage98}, uniqueness \cite{CarlierLachand-Robert01,FigalliKimMcCann11,MussaRosen78,RochetChone98} and robustness \cite{Basov05,FigalliKimMcCann11} have been well studied, among which the equivalence of function space convexity to the non-negative cross-curvature condition revealed by Figalli-Kim-McCann \cite{FigalliKimMcCann11} serves as a major milestone orienting our work. \medskip

	When parameterization of preferences is linear in agent types and price, Rochet and Chon\'e (1998, \cite{RochetChone98}) not only obtain existence results but also partially characterize optimal solutions and expound their economic interpretations, given that the monopolist profit can be characterized by {the aggregate difference between selling prices and quadratic manufacturing costs.}\medskip
	
	More generally, Carlier \cite{Carlier01} has proved existence results for general quasilinear utilities, where agent type and product type are not necessarily of the same dimension and the monopolist profit equals selling price minus some linear manufacturing cost.\medskip

	This paper generalizes the quasilinear case to the non-quasilinear case, which has many potential applications. For example, some fully nonlinear utilities include scenarios where agents are more sensitive to higher prices and where different agents might have different sensitivities to the same price. See Wilson~\cite[Chapter 7]{Wilson93} for the importance of taking income effects into account. 
	The generalized existence problem also appears as a conjecture by Basov \cite[Chapter~8]{Basov05}. However, only a few results are known for the multidimensional non-quasilinear case, and the impact of price on utility could be much more complicated.\medskip

	Recently, N\"oldeke-Samuelson \cite{NoldekeSamuelson15p} provided a general existence result assuming that the agent and product space are compact, by implementing a duality argument based on Galois connections. McCann-Zhang \cite{McCannZhang17} not only showed a general existence result assuming the single-crossing type condition and boundedness of the agent-type and product-type spaces, but generalized uniqueness and convexity results of Figalli-Kim-McCann \cite{FigalliKimMcCann11} to the non-quasilinear case, by using $G$-convexity arguments. In this paper, we also explore existence using $G$-convex analysis, which will be introduced in Subsection \ref{subsection:preliminary}, but with less restriction on boundedness of the product domain and without assuming the generalized single-crossing condition~\cite{McAfeeMcMillan88}. As a result of the lack of natural compactness, the proof of the existence result in this paper is entirely different from that in either of the earlier papers. It should be mentioned here that the existence results from this paper and the earlier two require no restrictions on the monopolist profit to take on a particular form, which is a generalization from much of the literature.\medskip

	The remainder of this paper is organized as follows. Section \ref{section:model} states the mathematical model and assumptions. It also introduces preliminaries, including $G$-convexity and $G$-subdifferentiability, and reformulates the monopolist's problem. In Section \ref{section:mainresult}, we state the existence theorem (Theorem \ref{maintheorem}) as well as the convergence results for sequences of $G$-convex functions (Proposition \ref{proposition:convergence}). Section \ref{section:futurework} proposes some directions for future work. We leave all the proofs of lemmas, propositions and the existence theorem in Section \ref{section:proofs}.

\bigskip

\section{Model}\label{section:model}

	Our model of the principal-agent problem is a bilevel optimization. After a monopolist publishes her price menu, each agent maximizes his utility through the purchase of at most one product. Knowing only the distribution of agent types, the monopolist maximizes aggregate profit based on agents' choices, which are based on the price menus.\medskip

	Suppose the agents' preferences are given by some parametrized utility function  $G(x, y, z)$, where $x$ is an $M$-dimensional vector of agent characteristics, $y$ is an $N$-dimensional vector of attributes of each product, and $z$ represents the price of each product. Denote {by} $X$ the space of agent types, by $Y$ the space of products, by $cl(Y)$ the closure of $Y$, by $Z$ the space of prices, and by $cl(Z)$ the closure of $Z$. In this paper, we only consider the case where both agent types and product attributes are continuous. \medskip

	The monopolist sells indivisible products to agents, i.e., she will sell neither a part/percentage of one product nor a product with some probability. Each agent buys at most one unit of product. 
	For any given price menu $p: cl(Y) \rightarrow cl(Z)$, an agent $x \in X$ knows his utility $G(x,y,p(y))$ for purchasing each product $y$ at price $p(y)$. It follows that each agent solves the following maximization problem 
	\begin{equation}\label{eqn_optimal_product}
		u(x):=\max_{y \in cl(Y)} G(x, y, p(y)),
	\end{equation}
	where $u(x)$ represents the maximal utility agent $x$ can obtain, and $u: X \rightarrow \R$ is also called the value function or indirect utility function. At this point, it is assumed that the maximum in \eqref{eqn_optimal_product} is attained for each agent $x$.
	\medskip

	If agent $x$ purchases product $y$ at price $p(y)$, the monopolist would earn from this transaction a profit of $\pi(x,y,p(y))$. For example, the monopolist profit can take the form $\pi(x,y,p(y)) = p(y)-c(y)$, where $c(y)$ is a variable manufacturing cost function. Summing over all agents in the distribution $d\mu(x)$, the monopolist's total profit is characterized by 
	\begin{equation}\label{eqn_monopolist_integral}
		\Pi(p, y):=\int_{X} \pi(x, y(x), p(y(x))) d\mu(x),
	\end{equation}
	which depends on her price menu $p: cl(Y) \rightarrow cl(Z)$ and  agents' choices $y: X \rightarrow cl(Y)$.\footnote{It is worth mentioning that in some literature, the monopolist's objective is to design a product line $\tilde{Y}$ (i.e.,~a subset of $cl(Y)$) and a price menu $\tilde{p}: \tilde{Y} \rightarrow \R$ that jointly maximize the overall monopolist profit. Then, given $\tilde{Y}$ and $\tilde{p}$, an agent of type $x$ chooses the product $y(x)$ that solves
	\begin{equation*}
		\max_{y \in \tilde{Y}} G(x,y, \tilde{p}(y)):= u(x).
	\end{equation*}
	Allowing the price to take value $\bar{z}$ (which may be $+\infty$), and assuming Assumption \ref{assmp:Gregular} below, the effect of designing a product line $\tilde{Y}$ and price menu $\tilde{p}: \tilde{Y}\rightarrow \R$ is equivalent to that of designing a price menu $p : cl(Y)\rightarrow (-\infty, +\infty]$, which equals $\tilde{p}$ on $\tilde{Y}$ and maps $cl(Y) \setminus \tilde{Y}$ to $\bar{z}$, such that no agents choose to purchase any product from $cl(Y) \setminus \tilde{Y}$, which is less attractive than the outside option $y_{\emptyset}$ according to Assumption \ref{assmp:Gregular}. In this paper, we use the latter as the monopolist's objective.
 \vspace{0.1cm}
 
	For any given price menu $p: cl(Y)\rightarrow (-\infty, +\infty]$, one can construct a mapping $y: X \rightarrow cl(Y)$ such that each $y(x)$ solves the maximization problem in \eqref{eqn_optimal_product}. But such mapping is not necessarily unique, without the single-crossing type assumptions. Therefore, we adopt in \eqref{eqn_monopolist_integral} the total profit as a functional of both price menu~$p$ and its corresponding mapping $y$.}\medskip

\begin{remark}
		One can replace the constraint
		\begin{equation}\label{product-p}
			y(x)\in \argmax_{y\in cl(Y)} G(x.y, p(y)),
		\end{equation}
		which is derived from equation \eqref{eqn_optimal_product}, by the incentive compatibility defined below, which are equivalent in the following sense: (1) all product-price pair satisfying \eqref{product-p} is incentive compatible; (2) for any incentive compatible product-price pair, there exists an (equivalent)  incentive compatible product-price pair which has the same value under the functional $\Pi$ and satisfies \eqref{product-p}.
\end{remark}

\begin{definition}[Incentive compatibility]
	A (product, price) pair of measurable mappings $(y,z): X \rightarrow cl(Y) \times cl(Z)$ on agent space $X$ is \textit{incentive compatible} if and only if $G(x,y(x),z(x)) \ge G(x, y(x'), z(x'))$ for all $(x,x')\in X^2$.
\end{definition}

	An incentive-compatible product-price pair $(y, p(y))$ ensures that no agent has the incentive to pretend to be another agent type.\medskip 

	In addition, we adopt a participation constraint in order to rule out the possibility of the monopolist charging exorbitant prices and the agents still having to make transactions despite this: each agent $x\in X$ will refuse to participate to the market if the maximum utility he can obtain is less than his reservation value $u_{\emptyset}(x)$, where the function $u_{\emptyset}: X \rightarrow \R$ is given in the form $u_{\emptyset}(x): = G(x, y_{\emptyset}, z_{\emptyset})$ for some $(y_{\emptyset}, z_{\emptyset}) \in cl(Y \times Z)$, where $y_{\emptyset}$ represents the outside option, whose price equals to some fixed value $z_{\emptyset} \in \R$ beyond the monopolist's control. \medskip

	For the monopolist profit, some literature assumes $\pi(x, y_{\emptyset}, z_{\emptyset}) \ge  0$ for all $x\in X$ to ensure that the outside option is harmless to the monopolist. Here, it is not necessary to adopt such an assumption for the sake of generality. \medskip

	The monopolist's problem can be described as follows:

	\begin{equation}\label{origin_problem}
		(P_0)
		\begin{cases}
			\sup \Pi(p,y)=\int_{X} \pi(x, y(x), p(y(x)))~ d\mu(x)\\
			s.t.\\
			(y,p(y)) \text{~is incentive compatible
				};\\
			 G(x, y(x), p(y(x))) \ge u_{\emptyset}(x) \text{ for all } x \in X;\\
			 p \text{  is lower semicontinuous}.
		\end{cases}
	\end{equation}

	We assume that $p$ is lower semicontinuous, without which the maximum in \eqref{eqn_optimal_product} may not be attained. We also rewrite the monopolist's problem in Proposition \ref{equiv_form}, which is an equivalent form of \eqref{origin_problem}.
	\medskip

	The purpose of the following subsections is to fix terminology and prepare the preliminaries for the main results in the next section.
	The proofs can be found in Section \ref{section:proofs}.
	\medskip

\subsection{$G$-Convex and $G$-Subdifferentiability}\label{subsection:preliminary}

	In this subsection, we introduce some tools from convex analysis and the notion of $G$-convexity (c.f. \cite{Trudinger14,Balder77,Singer97}), which is a generalization of ordinary convexity. Several results in this subsection echo those from Section 3 of McCann-Zhang \cite{McCannZhang17} but extend to unbounded domains.
	\medskip

	The following assumptions are made to the agents' direct utility $G$. We use $C^0(X)$ to denote the space of all continuous functions on $X$ and use $C^1(X)$ to denote the space of all differentiable functions on $X$ whose derivatives are continuous.  \medskip
 
 \begin{assumption}\label{assmp:Gregular}
 	Agents' utility $G \in C^{1}(cl(X\times Y \times Z))$, where the space of agents $X$ is a bounded open convex subset in $\R^M$ with $C^1$ boundary, the space of products $Y \subset \R^N$, and range of prices $Z=(\munderbar z,\bar z)$ with $-\infty < \munderbar z < \bar z \le +\infty$. {Assume $G(x,y,\bar{z}) := \lim_{z\rightarrow \bar{z}} G(x,y,z) \le G(x, y_{\emptyset}, z_{\emptyset})$ for all $(x,y) \in X \times cl(Y)$; and assume this inequality is strict when $\bar{z} = +\infty$.}
 \end{assumption}
 
	Here we do not necessarily assume $X$, $Y$, and $Z$ are compact spaces; in particular, $Y$ and $Z$ are potentially unbounded (i.e.,\ we do not set  \textit{a priori} bounds for product attributes or an \textit{a priori} upper bound for prices). However, we do specify a lower bound for the price range, since the monopolist has no incentive to set price close to negative infinity. The last inequality shows the highest price $\bar{z}$ is less preferred than the outside option.
	\medskip
 
\begin{assumption}\label{assmp:Gdecreasing}
 	$G(x,y,z)$ is strictly decreasing with respect to $z$ for each $(x,y) \in cl(X \times Y)$.
\end{assumption}
 
	This is to say that, for any given product, the higher the price paid to the monopolist, the lower the utility that will be left for the agent. For each $(x, y) \in X\times cl(Y)$ and $u\in G(x,y, cl(Z))$, define $H(x,y,u) := z$ whenever $G(x,y,z) = u$, i.e., $H(x, y, \cdot)= G^{-1}(x,y,\cdot)$. Therefore, $H(x,y,u)$ represents the price paid by agent $x$ for product $y$ when receiving value $u$.
	\medskip
 
	Proposition \ref{lemma_continuity} shows that the inverse function of $G$ is also continuous because $G$ is continuous and monotonic on the price variable. 
	\medskip
 
\begin{proposition}\label{lemma_continuity}
 	Given Assumption \ref{assmp:Gregular} and Assumption \ref{assmp:Gdecreasing}, the function $H$ is continuous.
\end{proposition}

	Recall that the subdifferential of a function $u$ at $x_0$ is defined as the set:
	\begin{equation*}
		\partial u(x_0) = \left\{ y \in cl(Y)|~ u(x) - u(x_0) \ge \langle  x- x_0,  y \rangle \text{ for all } x \in X  \right\}.
	\end{equation*}
	Here we use $ \langle , \rangle$ to denote the Euclidean inner product. This set is nonempty for every $x_0$ if and only if $u$ is convex. Then for any convex function $u$ on $X$ and any fixed point $x_0 \in X$, there exists $y_0 \in \partial u(x_0)$, satisfying
	\begin{equation}\label{convex_function}
		u(x) \ge  \langle x , y_0\rangle -\left(\langle x_0, y_0\rangle -  u(x_0)\right)	\text{  for all $x \in X$},
	\end{equation} 
	where equality holds at $x = x_0$. On the other hand, if for any $x_0\in X$, there exists $y_0$ such that \eqref{convex_function} holds for all $x\in X$, then $u$ is convex. The following definition is analogous to this property, which is a special case of $G$-convexity, when $G(x,y,z) = \langle x, y \rangle -z$.

\begin{definition}[$G$-convexity]
	A function $u\in C^0(X)$ is called {\it $G$-convex} if for each $x_0 \in X$, there exist $y_0 \in   cl(Y)$ and $z_0 \in  cl(Z)$ such that $u(x_0)=G(x_0, y_0, z_0)$ and $u(x)\ge G(x, y_0, z_0)$ for all $x\in X$.
\end{definition}

	Similarly, one can also generalize the definition of subdifferential from \eqref{convex_function}.

\begin{definition}[$G$-subdifferentiability]
	The $G$-subdifferential of a function $u: X \rightarrow \R$ is a point-to-set mapping defined by
	\begin{equation*}
	\partial^G u(x):= \left\{ y\in  cl(Y)|~ H(x,y,u(x)) \text{ is defined and }u(x')\ge G(x',y, H(x,y,u(x))) \text{ for all } x'\in X \right\}.
	\end{equation*}
	A function $u$ is said to be {\it $G$-subdifferentiable} at $x$ if and only if $\partial^G u(x) \neq \emptyset$.\footnote{In Trudinger \cite{Trudinger14}, this point-to-set mapping $\partial^G u$ is also called $G$-$normal$ mapping; see this paper for more properties related to $G$-convexity.}
\end{definition}

	In particular, if $G(x,y,z) = \langle x, y \rangle - z$, then the $G$-subdifferential coincides with the subdifferential. There are other generalizations of convexity and subdifferentiability. For instance, $h$-convexity in Carlier \cite{Carlier01}, or equivalently, $b$-convexity in Figalli-Kim-McCann \cite{FigalliKimMcCann11}, or $c$-convexity in  Gangbo-McCann \cite{GangboMcCann96}, is a special form of $G$-convexity, where $G(x,y,z)=  h(x,y) -z$, which plays an important role in the quasilinear case. For more references of convexity generalizations, see Kutateladze-Rubinov \cite{KutateladzeRubinov72}, Elster-Nehse \cite{ElsterNehse74}, Balder \cite{Balder77}, Dolecki-Kurcyusz \cite{DoleckiKurcyusz78},  Singer \cite{Singer97},  Rubinov \cite{Rubinov00a}, and Martínez-Legaz \cite{MartinezLegaz05}.\medskip

	As mentioned above, a well-known result in convex analysis is that a function is convex if and only if it is subdifferentiable everywhere. The following lemma adapts this to $G$-convexity. 

\begin{lemma}\label{convex-subdiff}
	Given Assumption \ref{assmp:Gdecreasing}, a function $u: X \rightarrow \R$ is $G$-convex if and only if it is $G$-subdifferentiable everywhere.
\end{lemma}

	Using Lemma \ref{convex-subdiff}, one can show the following result, which connects incentive compatibility in the economic context with $G$-convexity and $G$-subdifferentiability in mathematical analysis, generalizing the results of Rochet \cite{Rochet87} and Carlier \cite{Carlier01}.

\begin{lemma}[$G$-convex utilities characterize incentive compatibility]\label{incen/convex}
	Let $(y,z)$ be a pair of mappings from $X$ to $cl(Y) \times cl(Z)$. Given Assumption \ref{assmp:Gdecreasing}, this (product, price) pair is incentive compatible if and only if $u(\cdot):=G(\cdot,y(\cdot),z(\cdot))$ is $G$-convex and $y(x)\in \partial^G u(x)$ for each $x \in X$.
\end{lemma}

\subsection{Implementability}\label{subsection:implementability}
We introduce implementability here, which is closely related to incentive-compatibility and can also be exhibited by $G$-convexity and $G$-subdifferential. 
	
\begin{definition}[Implementability]
		A function $y: X \rightarrow cl(Y)$ is called \textit{implementable} if and only if there exists a function $z: X \rightarrow \R$  such that the pair $(y, z)$ is incentive compatible.
\end{definition}
	
\begin{remark}\label{rmk:implementability}
		Allowing Assumption \ref{assmp:Gdecreasing}, a function $y$ is implementable if and only if there exists a price menu $p: cl(Y) \rightarrow \R$ such that the pair $(y, p(y))$ is incentive compatible.
\end{remark}
	
	As a corollary of Lemma \ref{incen/convex},  implementable functions can be characterized as $G$-subdifferential of $G$-convex functions. 
	
\begin{corollary}[$G$-convex utilities characterize implementability]\label{cor:implementable}
	Given Assumption \ref{assmp:Gdecreasing}, a function $y: X \rightarrow cl(Y)$ is implementable if and only if there exists a $G$-convex function $u(\cdot)$ such that $y(x) \in \partial^G u(x)$ for each $x\in X$.
\end{corollary}

	When parameterization of preferences is linear in agent types and price, Corollary \ref{cor:implementable} says that a function is implementable if and only if it is monotone increasing. In general quasilinear cases, this coincides with Proposition 1 of Carlier \cite{Carlier01}. 
	\medskip

\subsection{Reformulation of the Monopolist's Problem}\label{subsection:reformulation}
From the original monopolist's problem \eqref{origin_problem}, we replace product-price pair $(p,y)$ by the value-product pair $(u,y)$, using $u(\cdot) = G(\cdot, y(\cdot), p(y(\cdot)))$. Combining this with Lemma \ref{incen/convex}, the incentive-compatibility constraint $(y,p(y))$ is equivalent to $G$-convexity of $u(\cdot)$ and $y(x) \in \partial^G u(x)$ for all $x\in X$. Therefore, one can rewrite the monopolist's problem as follows.

\begin{proposition}[Reformulation of the monopolist’s problem]\label{equiv_form}
	Given Assumptions \ref{assmp:Gregular} and \ref{assmp:Gdecreasing}, the monopolist's problem $(P_0)$ is equivalent to
	\begin{equation}\label{Principal_new_problem}
		(P)
		\begin{cases}
			\sup \tilde{\Pi}(u,y):=\int_{X} \pi(x, y(x), H(x,y(x), u(x))) d\mu(x)\\
			s.t.\\
			 $u$ \text{ is } $G$\text{-convex };\\
			 y(x) \in \partial^G u(x) \text{ and }u(x)\ge u_{\emptyset}(x) \text{ for all } x \in X.
		\end{cases}
	\end{equation}
\end{proposition}

\medskip

\subsection{Other Assumptions}\label{subsection:assumptions}
	In the next section, we will show the existence result of the rewritten monopolist's problem $(P)$ given in \eqref{Principal_new_problem}. For the preparation of the main result, we introduce the following assumptions. Note that, even in the one-dimensional case, we assume no single-crossing type condition. We also include two propositions here, which will be employed in the proofs of Proposition \ref{proposition:convergence} and the existence theorem. 
\medskip

\begin{assumption}\label{assmp:Gcoordinate-monotone}
	$G$ is coordinate-monotone in $x$. That is, for each $(y,z)\in cl(Y\times Z)$ and all $ (\alpha, \beta) \in X^2$, if $\alpha_i\ge \beta_i$ for all $ i=1,2,...,M$, then $G(\alpha,y,z)\ge G(\beta, y,z)$.
\end{assumption}

	In Assumption \ref{assmp:Gcoordinate-monotone}, we assume that agent utility increases along each agent attribute coordinate. Given coordinate monotonicity of $G$ in the first variable, one can show that all the $G$-convex functions are nondecreasing. Therefore, the value functions are also monotonic with respect to the agent attributes.
	\medskip

\begin{proposition}\label{nondecreasing}
	Given Assumption \ref{assmp:Gcoordinate-monotone}, $G$-convex functions are nondecreasing in coordinates.
\end{proposition}

	In the following, we use  $D_x G(x,y,z) := \left(\frac{\partial G}{\partial x_1}, \frac{\partial G}{\partial x_2}, \dots, \frac{\partial G}{\partial x_M}\right)(x,y,z)$ to denote derivative with respect to $x$. For any vector in $\R^M$ or $\R^N$, we use $\norm{\cdot}$ and $\norm{\cdot}_{\alpha}$ to denote its Euclidean  $2$-norm and $\alpha$-norm ($\alpha \ge 1$), respectively. For example, for $x\in \R^M$, we have $\norm{x} = \sqrt{\sum_{i=1}^{M} x_i^2}$ and $\norm{x}_{\alpha} = \left(\sum_{i=1}^{M} |x_i|^{\alpha}\right)^{\frac{1}{\alpha}}$. We say a function $f: X \times Y \times Z \rightarrow \R$ is asymptotically decreasing super-linearly (respectively, linearly) with respect to the variable $y\in Y$ if and only if $\frac{f(x,y,z)}{\norm{y}} \rightarrow -\infty ~(\text{respectively, }C(x,z))$ as $\norm{y}\rightarrow +\infty$, where $C(x,z)$ is a negative variable depending only on $x$ and $z$. We say a function $f: Y\rightarrow \R$ is asymptotically super-linear if and only if $\frac{f(y)}{\norm{y}} \rightarrow +\infty$ as $\norm{y}\rightarrow +\infty$.\medskip
 
	In Rochet-Chon\'e's model, $H(x,y,u) = x\cdot y -u$ and $\pi(x,y,z) = z-c(y)$ with a quadratic cost function $c$. In this case, $\pi(x,y,H(x,y,u)) = x\cdot y -u -c(y)$. Since $c$ is asymptotically super-linear and the space $X$ is bounded, it is reasonable to assume the following:\medskip

\begin{assumption}\label{assmp:Gtech0}
	$\pi(x,y,H(x,y,u))$ is {asymptotically decreasing (super-)linearly with respect to $y$ and  asymptotically decreasing at least linearly with respect to $u$}. More precisely, there exist $\alpha \ge 1$, $a_1, a_2> 0$, and $b\in \R$ such that $\pi(x,y,H(x,y,u)) \le -a_1 \norm{y}_{\alpha}^{\alpha} - a_2 u +b$ for all $ (x, y, u)\in \left\{ (x,y, G(x,y,z))|~ x\in X, y\in cl(Y), z\in \R \right\}$; or equivalently, $\pi(x,y,z) + a_2 G(x,y,z) \le -a_1 \norm{y}_{\alpha}^{\alpha}  +b$ for all $ (x, y, z)\in X\times cl(Y)\times \R$.
\end{assumption}
	
	As shown in the alternative formulation, Assumption \ref{assmp:Gtech0} requires the existence of some weighted surplus which is asymptotically decreasing super-linearly with respect to the product (when $\alpha >1$). In the case where $Y$ is bounded, Assumption \ref{assmp:Gtech0} is equivalent to the existence of some weighted surplus bounded from above.   
	\medskip

	Assumptions \ref{assmp:Gtech1} - \ref{assmp:Gtech3} are some technical assumptions on $D_xG$, which are automatically satisfied for $X$, $Y$, and $Z$ being bounded.
	\medskip

\begin{assumption}\label{assmp:Gtech1}
	$D_x G(x,y,z)$ is Lipschitz with respect to $x$, uniformly in $(y,z)$, meaning there exists $k\in \R$ such that $\norm{D_xG(x,y,z)-D_x G(x',y,z)}\le k\norm{x-x'}$ for all $(x, x',y, z)\in X^2\times cl(Y) \times cl(Z)$.
\end{assumption}

\begin{assumption}\label{assmp:Gtech2}
	There exist $ \beta \in (0, \alpha]$, $c>0$, and $ d\in \R$ such that $\norm{D_x G(x,y,z)}_{1}\le c\norm{y}_{\beta}^{\beta} +d$ for all $ (x, y, z)\in X\times cl(Y) \times cl(Z)$.
\end{assumption}

\begin{assumption}\label{assmp:Gtech3}
	 Coercivity of {$1$-norm} of $(D_xG)$. For all $ s>0$, there exists $r>0$ such that $\sum_{i=1}^{M} |D_{x_i}G(x,y,z)|\ge s$ for all $(x, y, z)\in X\times  cl(Y) \times cl(Z)$, whenever $\norm{y}\ge r$.
\end{assumption}

	Allowing Assumption \ref{assmp:Gcoordinate-monotone}, {the derivatives $D_{x_i}G$ are always nonnegative; therefore,} we no longer need to take absolute values of $D_{x_i}G$ in the inequality of Assumption \ref{assmp:Gtech3}. Then Assumption~\ref{assmp:Gtech3} says that the marginal utility of agents who select the same product $y$ will increase to infinity as $\norm{y}$ approaches infinity, uniformly for all agents and prices. For instance, when $M = N$, utility $G(x,y,z) = \sum_{i=1}^{M} x_iy_i^2 -f(y,z) $, with $f \in C^1(Y \times Z)$ strictly increasing with respect to $z$, satisfies Assumption~\ref{assmp:Gtech3}, because 
	$$\sum_{i=1}^{M} \left|D_{x_i}G(x,y,z)\right| = \sum_{i=1}^{M} D_{x_i}G(x,y,z) = \sum_{i=1}^{M} y_i^2 \rightarrow +\infty$$ 
	as $\norm{y} \rightarrow + \infty$. In addition, with appropriate $\pi$, this $G$ could satisfy all the other assumptions. For instance, if we take $\bar{z} = +\infty$, then profit $\pi(x,y,z) = z - \sum_{i=1}^{M} y_i^4$ and agents' utility $G(x,y,z)= \sum_{i=1}^{M} x_iy_i^2 - z^2 $ together satisfy Assumption \ref{assmp:Gregular} -  \ref{assmp:Pi1}.
	\medskip

	In general, if $Y$ is bounded, any $G$ in the form of $G(x,y,z) = b(x,y) - f(y,z)$, with $b \in C^1(cl(X\times Y))$ and $f\in C^0(cl(Y \times Z))$, satisfies Assumption~\ref{assmp:Gtech1} - \ref{assmp:Gtech3}. This class of separable utility functions for the agent are also considered in Nöldeke-Samuelson \cite{NoldekeSamuelson15p}, where they serve as the prime example for utility functions that are not necessarily quasilinear but yield strong implementability. 
	\medskip

	Proposition \ref{Subdiff/Bdd} presents that uniform boundedness of the agents' value functions on some compact subset implies uniform boundedness of the corresponding agents' choices of their favorite products. 
	\medskip

\begin{proposition}[{Uniform boundedness of $G$-convex functions on a compact set implies that of $G$-subdifferentials}]\label{Subdiff/Bdd}
	Given Assumptions \ref{assmp:Gregular}, \ref{assmp:Gdecreasing}, \ref{assmp:Gcoordinate-monotone}, \ref{assmp:Gtech3}, and let $u(\cdot)$ be a $G$-convex function on $X$, $\omega$ be a compact subset of $X$, $\delta>0, R>0$, satisfying $\omega+\delta\overline{B(0,1)}\subset X$ and $|u(x)|\le R$ for all $x\in \omega + \delta \overline{B(0,1)}$ (here, $\overline{B(0,1)}$ denotes the closed unit ball of $\R^M$). Then, there exists $T = T(\omega,\delta, R) > 0$ such that $\norm{y}\le T$ for any $x \in \omega$ and any $y\in \partial^Gu(x)$.
\end{proposition}

	Assumptions~\ref{assmp:Pi1} states constraints on continuity of the principal's profit function $\pi$, integrability of participation constraint $u_{\emptyset}$. From now on, we assume $\mu$ to be equivalent to the Lebesgue measure. One might seek to extend the corresponding results to more general settings, e.g., assuming $\mu$ to be absolutely continuous with respect to the Lebesgue measure.

\begin{assumption}\label{assmp:Pi1}
	The profit function $\pi$ is continuous on $cl(X\times Y\times Z)$. The participation constraint $u_{\emptyset}$ is integrable with respect to $d \mu$, where the measure  $\mu$ is equivalent to the Lebesgue measure restricted on $X$.
\end{assumption}

	For $\alpha \ge 1$, denote $L^{\alpha}(X)$ as the space of measurable functions for which the $\alpha$-th power of the absolute value is Lebesgue integrable on $X$. That is, a function $f: X\rightarrow \R$ is in $L^{\alpha}(X)$ if and only if 
	$$\norm{f}_{L^{\alpha}(X)}:= \left(\int_{X} |f(x)|^{\alpha} dx \right)^{\frac{1}{\alpha}} <+\infty,$$
	where $\norm{f}_{L^{\alpha}(X)}$ is the associated $L^{\alpha}$ norm of $f$ on $X$.  
	For instance, Assumption \ref{assmp:Pi1} implies $u_{\emptyset}\in L^{1}(X)$.  For any function $f$ on $X$, its {\em essential supremum} is defined as:
	\begin{equation}
		\norm{f}_{L^{\infty}(X)} := \inf\{C\ge 0: |f(x)| \le C \text{ for Lebesgue almost every } x \in X \}.
	\end{equation}
	Denote by $L^{\infty}(X)$ the space of measurable functions on $X$ whose essential supremum is finite. Also, denote by $L_{loc}^{\infty}(X)$ the space of measurable functions on $X$ which belong to $L^{\infty}(K)$ for all compact subsets $K$ of $X$. 
	\medskip

\bigskip

\section{Main result}\label{section:mainresult}
	In this section, we state the existence theorem, the proof of which is provided in Section \ref{section:proofs}. 

\begin{theorem}[Existence]\label{maintheorem}
	Under Assumptions \ref{assmp:Gregular} - \ref{assmp:Pi1}, the monopolist's problem $(P)$ admits at least one solution.
\end{theorem}

	Technically, in order to demonstrate existence, we start from a sequence of value-product pairs, whose total profits have a limit that is equal to the supremum of $(P)$.
	Then we need to show that this sequence converges, up to a subsequence, to a pair of limit mappings. Then we show this limit value-product pair satisfies the constraints of $(P)$, and its corresponding total payoff is no worse than those of any other admissible pairs. \medskip

	Suppose the space of products $Y$ is bounded, then Assumption \ref{assmp:Gtech0} - \ref{assmp:Gtech3} could be simplified, and Proposition \ref{Subdiff/Bdd} holds automatically. Besides, some steps in the proof of the main theorem would be simplified. If, in addition, both spaces $Y$ and $Z$ are bounded, Assumption \ref{assmp:Gtech0} - \ref{assmp:Gtech3} are automatically satisfied and the proofs will be much simpler.
	\medskip

	In the following, we denote by $W^{1,1}(X)$  the Sobolev space of $L^1$ functions whose first derivatives exist in the weak sense and belong to $L^1(X)$. 
	For any function $f: X \rightarrow \R$, define its $W^{1,1}(X)$ norm as 
	\begin{equation*}
		\norm{f}_{W^{1,1}(X)} = \norm{f}_{L^1(X)} + \sum_{i=1}^{M}\norm{D_{x_i}f}_{L^1(X)}.
	\end{equation*}
	For more properties of Sobolev spaces and weak derivatives, see Evans~\cite[Chapter 5]{Evans98}. If $\omega$ is some open subset of $X$, the notation $\omega \subset \subset X$ means that the closure of $\omega$ is also included in $X$.
	\medskip

	Lemma \ref{lemma1} provides convergence results for sequences of convex functions, which are uniformly bounded in Sobolev spaces on open convex subsets. We state this classical result without proof, which can be found in Carlier \cite{Carlier01}.
	\medskip

\begin{lemma}\label{lemma1}
	Let $\{u_n\}_{n\in \N}$ be a sequence of convex functions on $X$ such that for every open convex set $\omega \subset \subset X$, the following holds:
	\begin{equation*}
	\sup\limits_{n} \norm{u_n}_{W^{1,1}(\omega)} < +\infty.
	\end{equation*}
	Then there exists a function $u^*$ which is convex in $X$, a measurable subset $A$ of $X$ and a subsequence again labeled $\{u_n\}_{n\in \N}$ such that\\
	1. $\{u_n\}_{n\in \N}$ converges to $u^*$ uniformly on compact subsets of $X$;\\
	2. $\{D u_n\}_{n\in \N}$ converges to $D u^*$ pointwise in $A$ and $\dim_{H}(X\setminus A)\le M-1$, where $\dim_{H}(X\setminus A)$ is the Hausdorff dimension of $X\setminus A$.
\end{lemma}

	We extend the above convergence result to $G$-convex functions in the following proposition, which is required in the proof of the existence theorem, as it extracts a limit function from a converging sequence of value functions.

\begin{proposition}\label{proposition:convergence}
	Assume Assumptions \ref{assmp:Gregular}, \ref{assmp:Gdecreasing}, \ref{assmp:Gcoordinate-monotone}, \ref{assmp:Gtech1}, \ref{assmp:Gtech3}, and let $\{u_n\}_{n\in \N}$ be a sequence of $G$-convex functions on $X$ such that for every open convex set $\omega \subset \subset X$, the following holds:
	\begin{equation*}
	\sup\limits_{n} \norm{u_n}_{W^{1,1}(\omega )} < +\infty.
	\end{equation*}
	Then there exists a function $u^*$ which is  $G$-convex in $X$, a measurable subset $A$ of $X$, and a subsequence again labeled $\{u_n\}_{n\in \N}$ such that\\
	1. $\{u_n\}_{n\in \N}$ converges to $u^*$ uniformly on compact subsets of $X$;\\
	2. $\{D u_n\}_{n\in \N}$ converges to $D u^*$ pointwise in $A$ and $\dim_{H}(X\setminus A)\le M-1$.
\end{proposition}

	In the proof of Proposition \ref{proposition:convergence}, we show that the sequence of $G$-convex functions is convergent by applying results from Lemma \ref{lemma1}, then prove that the limit function is also $G$-convex. 
	\medskip

\bigskip

\section{Future Work}\label{section:futurework}

	We have strong interest in giving an explicit solution for the non-quasilinear example on the real line and in high dimension. We also would like to investigate, among other things, the conditions under which the matching map $y: X \rightarrow cl(Y)$ is continuous and/or differentiable. Given the technical arguments employed in this paper, it may be very fruitful to study possible generalizations of other known results for convex functions to $G$-convex functions.
	\medskip

\bigskip

\section{Proofs}\label{section:proofs}

\begin{proof}[Proof of Proposition \ref{lemma_continuity}]
	(Proof by contradiction). Suppose $H$ is not continuous, then there exists a sequence $\{(x_n, y_n, z_n)\}_{n\in \N} \subset cl(X\times Y \times Z)$ converging to $(x, y, z)$ and $\varepsilon >0$ such that 
	$$|H(x_n, y_n, z_n) - H(x,y,z)|>\varepsilon \text{ for all $n\in \N$}.$$ 
	Without loss of generality, we assume $H(x_n, y_n, z_n) - H(x,y,z)>\varepsilon$ for all $n\in \N$. Therefore, we have $H(x_n, y_n, z_n) > H(x,y,z)+\varepsilon$. By Assumption \ref{assmp:Gdecreasing}, this implies 
	$$z_n < G(x_n, y_n, H(x,y,z)+\varepsilon) \text{ for all $n\in \N$.}$$ 
	Taking limit $n\rightarrow \infty$ at both sides, since $G$ is continuous from Assumption \ref{assmp:Gregular}, we have 
	$$z \le G(x, y, H(x,y,z)+\varepsilon).$$ This implies $H(x,y,z) \ge H(x,y,z)+\varepsilon$, a contradiction.
\end{proof}

\vspace{0.3cm}

\begin{proof}[Proof of Lemma \ref{convex-subdiff}]
	Assume $u$ is $G$-convex, we want to show that $u$ is $G$-subdifferentiable everywhere, i.e., we need to prove  $\partial^G u(x_0)\neq \emptyset$ for all $x_0\in X$.
	
	Since $u$ is $G$-convex, by definition, for each $x_0$, there exists $ y_0, z_0$ such that $u(x_0) = G(x_0,y_0,z_0)$ and
	\begin{equation*}
	u(x)\ge G(x, y_0, z_0) = G(x, y_0, H(x_0,y_0,u(x_0))) \text{  for all  $ x \in X$}.
	\end{equation*}
	By the definition of $G$-subdifferentiability, $y_0 \in \partial^G u(x_0)$, i.e., $\partial^G u(x_0) \neq \emptyset$.
	
	On the other hand, assume $u$ is $G$-subdifferentiable everywhere, then for each $ x_0 \in X$,  there exists $ y_0 \in \partial^G u(x_0)$. Set $z_0:=H(x_0,y_0,u(x_0))$ so that $u(x_0) = G(x_0, y_0, z_0)$.
	
	Since $y_0\in \partial^G u(x_0)$, we have 
	\begin{equation*}
	u(x)\ge G(x,y_0,H(x_0,y_0,u(x_0))) = G(x,y_0,z_0) \text{  for all $x\in X$ }.
	\end{equation*}
	By definition, $u$ is $G$-convex.
\end{proof}

\vspace{0.3cm}

\begin{proof}[Proof of Lemma \ref{incen/convex}]
	$``\Rightarrow"$. Suppose $(y,z)$ is incentive compatible. For any fixed $x_0 \in X$, let $y_0 = y(x_0)$ and $z_0 = z(x_0)$. Then 
	$$u(x_0) = G(x_0, y(x_0), z(x_0)) = G(x_0, y_0, z_0).$$ 
	By incentive compatibility of the contract $(y,z)$, one has 
	$$G(x, y(x), z(x)) \ge G(x, y(x_0), z(x_0)) \text{ for any $x\in X$.}$$ This implies $u(x) \ge G(x,y_0,z_0)$ for any $x\in X$, because $u(x)= G(x, y(x), z(x))$,  $y_0 = y(x_0)$, and $z_0 = z(x_0)$. By definition, $u$ is $G$-convex. 
	
	Since $u(x_0)=G(x_0, y_0, z_0)$, by definition of function $H$ one has $z_0 = H(x_0, y_0, u(x_0))$.  Combining with $u(x) \ge G(x, y_0, z_0)$ for any $x\in X$, which is concluded from above, we have 
	$$u(x)\ge G(x, y_0, H(x_0, y_0, u(x_0))) \text{ for any $x\in X$.}$$ 
	By definition of  $G$-subdifferentiability, one has $y_0 \in \partial^G u(x_0)$. Therefore, $y(x_0) = y_0 \in \partial^G u(x_0)$.
	
	$``\Leftarrow"$. Assume that $u = G(x, y(x),z(x))$ is $G$-convex and $y(x)\in \partial^G u(x)$ for any $x\in X$. For any fixed $x \in X$, since $y(x)\in \partial^G u(x)$, one has 
	\begin{equation}\label{eqn_prop3.4}
		u(x')\ge G(x', y(x), H(x, y(x), u(x))) \text{ for any $x'\in X$.}
	\end{equation} 
	 
	Since $u(x) = G(x, y(x),z(x))$, by definition of function $H$, one has $z(x) = H(x,y(x), u(x))$. Combined with the inequality \eqref{eqn_prop3.4}, we have $$u(x')\ge G(x', y(x), z(x)) \text{ for any $x'\in X$.}$$ Notice $u(x') = G(x',y(x'),$ $z(x')) $. Thus, 
	$$G(x',y(x'),z(x')) = u(x') \ge G(x', y(x), z(x)) \text{ for any $x, x'\in X$.}$$
	By definition, $(y,z)$ is incentive compatible.
\end{proof}

\vspace{0.3cm}

\begin{proof}[Proof of Remark \ref{rmk:implementability}]
	One direction is easier: given $p$ and $y$, define $z(\cdot):= p(y(\cdot))$. Then the conclusion follows directly. \medskip
	
	Given an incentive-compatible pair $(y, z): X \rightarrow cl(Y) \times \R$, we need to construct a price menu $p: cl(Y)\rightarrow \R$. If $y= y(x)$ for some $x\in X$, define $p(y):= z(x)$; for any other $y \in cl(Y)$, define $p(y) := \bar{z}$. \medskip
	
	We first show $p$ is well-defined. Suppose $y(x) = y(x')$ with $x\ne x'$, from incentive compatibility of $(p,y)$, we have 
	$$G(x,y(x), z(x)) \ge G(x, y(x'), z(x')) = G(x, y(x), z(x')).$$ Since $G$ is strictly decreasing with respect to its third variable, the above inequality implies $z(x) \le z(x')$. Similarly, one has $z(x) \ge z(x')$. Therefore, $z(x) = z(x')$ and thus $p$ is well-defined. \medskip
	
	The incentive compatibility of $(y, p(y))$ follows from that of $(y, z)$ and definition of $p$.
\end{proof}

\vspace{0.3cm}

\begin{proof}[Proof of Corollary \ref{cor:implementable}]
	One direction is immediately derived from the definition of implementability and Lemma \ref{incen/convex}.
	\medskip
	
	Suppose there exists some convex function $u$ such that $y(x) \in \partial^G u(x)$ for each $x \in X$. Define $z(\cdot):= H(\cdot, y(\cdot), u(\cdot))$, then $u(x) = G(x, y(x), z(x))$.
	Lemma \ref{incen/convex} implies $(y, z)$ is incentive compatible, and thus $y$ is implementable.
\end{proof}

\vspace{0.3cm}

\begin{proof} [Proof of Proposition \ref{equiv_form}] 
	We need to prove both directions for equivalence of $(P_0)$ and $(P)$.
	\medskip
	
	1. For any incentive-compatible pair $(y, p(y))$, define $u(\cdot) := G(\cdot,y(\cdot), p(y(\cdot)))$. Then by Lemma \ref{incen/convex}, we have $u(\cdot)$ is $G$-convex and $y(x) \in \partial^G u(x)$ for all $x \in X$. From the participation constraint, 
	$$G(x, y(x), p(y(x))) \ge u_{\emptyset}(x) \text{ for all $x\in X$. }$$
	This implies $u(x)\ge u_{\emptyset}(x)$ for all $x\in X$. Besides, two integrands are equal: $\pi(x, y(x), p(y(x))) = \pi(x,y(x), H(x,y(x), u(x)))$. Therefore, $(P_0) \le (P)$.
	\medskip
	
	2. On the other hand, assume $u(\cdot)$ is $G$-convex, $y(x)\in \partial^G u(x)$ and $u(x) \ge u_{\emptyset}(x)$ for all $x \in X$. From Corollary \ref{cor:implementable} and Remark \ref{rmk:implementability}, we know $y$ is implementable and there exists a price menu $p: cl(Y) \rightarrow \R$ such that the pair $(y, p(y))$ is incentive compatible, where
	\begin{equation*}
		p(y) =
		\begin{cases}
			H(x,y(x), u(x)) & \text{ if } y = y(x) \in y(X) :=\{ y(x) \in cl(Y) |~ x \in X \};\\
			\bar{z} & \text{ otherwise}.\\
		\end{cases}
	\end{equation*}
	
	Firstly, the mapping $p$ is well-defined, using the same argument as that in Remark \ref{rmk:implementability}.
	
	Secondly, the participation constraint holds since 
	$$G(x,y(x), p(y(x))) = u(x) \ge u_{\emptyset}(x) \text{ for all $x\in X$.}$$ 

	Thirdly, let us show this price menu $p$ is lower semicontinuous. Let $\tilde{p}$ be the restriction of $p$ to $y(X)$. Suppose that $\{y_n\}_{n\in \N} \subset y(X)$ converges $y_0 \in y(X)$ with $y_n = y(x_n)$ and $y_0 = y(x_0)$, satisfying $$\lim\limits_{n \rightarrow \infty} \tilde{p}(y_n) = \liminf\limits_{y \rightarrow y_{0}} \tilde{p}(y).$$  
	Let $z_n:= \tilde{p}(y_n)$ and $z_{\infty}:=\lim\limits_{n \rightarrow \infty} z_n$. To prove lower semicontinuity of $\tilde{p}$, we need to show $\tilde{p}(y_0)\le z_{\infty}$.  Since $y_n \in \partial^G u(x_n)$, we have $$u(x) \ge G(x,y_n, H(x_n, y_n, u(x_n))) = G(x, y_n, z_n).$$ Taking $n\rightarrow \infty$, this implies $u(x)\ge G(x, y_0, z_{\infty})$. Therefore, 
	$$G(x_0, y_0, \tilde{p}(y_0)) = u(x_0) \ge G(x_0, y_0, z_{\infty}).$$ By Assumption \ref{assmp:Gdecreasing}, we know $\tilde{p}(y_0) \le z_{\infty}$. Thus $\tilde{p}$ is lower semicontinuous. Since $p$ is an extension of $\tilde{p}$ from $y(X)$ to $cl(Y)$ as its lower semicontinuous hull, satisfying $v(y)= \bar{z}$ for all $y\in cl(Y)\setminus y(X)$, we know $p$ is also lower semicontinuous.
	\medskip
	
	Lastly, two integrands are equal: $\pi(x, y(x), p(y(x))) = \pi(x,y(x), H(x,y(x), u(x)))$. Therefore, $(P_0) \ge (P)$.
\end{proof}

\vspace{0.3cm}

\begin{proof}[Proof of Proposition \ref{nondecreasing}]
	Let $u$ be any $G$-convex function, and let $\alpha$, $\beta$ be any two agent types in $X$ with $\alpha \ge \beta$. By $G$-convexity of $u$, for this $\beta$, there exist $y\in cl(Y)$ and $z \in cl(Z)$ such that 
	$$u(\beta)=G(\beta, y,z) \text{ and } u(x)\ge G(x, y,z) \text{ for any $x\in X$.}$$ 
	
	Since $\alpha \ge \beta$, by Assumption \ref{assmp:Gcoordinate-monotone}, we have $G(\alpha, y,z)\ge G(\beta,y,z)$. Combining with $u(\alpha)\ge G(\alpha, y,z)$ and $u(\beta) = G(\beta,y,z)$, one has $u(\alpha) \ge u(\beta)$. Thus, $u$ is nondecreasing.
\end{proof}

\vspace{0.3cm}

\begin{proof}[Proof of Proposition \ref{Subdiff/Bdd}]
	(Proof by contradiction).\medskip
	
	By Assumption \ref{assmp:Gcoordinate-monotone} and Assumption \ref{assmp:Gtech3}, for $s=\frac{4R\sqrt{M}}{\delta}$, there exists $r>0$ such that for any $(x, y, z)\in X \times  cl(Y) \times cl(Z)$, whenever $\norm{y}\ge r$, we have $$\sum\limits_{i=1}^{M}D_{x_i}G(x,y,z)\ge \frac{4R\sqrt{M}}{\delta}.$$
	
	Assume the boundedness conclusion of this proposition is not true. Then for this $r$, there exist $ x_0 \in \omega$ and  $ y_0\in \partial^G u(x_0)$ such that $\norm{y_0}\ge r$. Thus,
	\begin{equation}\label{eqn_coercivity}
		\sum\limits_{i=1}^{M}D_{x_i}G(x,y_0,z)\ge \frac{4R\sqrt{M}}{\delta} \ \ \text{ for all } x \in X \text{ and } z \in \R.
	\end{equation}
	Since $y_0 \in \partial^G u(x_0)$, by definition of the $G$-subdifferential, we have 
	$$u(x)\ge G(x,y_0,H(x_0,y_0,u(x_0))) \text{ for any $ x \in X$.}$$ 
	
	Take $x=x_0+\delta x_{-1}$, where $x_{-1}:=(\frac{1}{\sqrt{M}}, \frac{1}{\sqrt{M}}, \cdots, \frac{1}{\sqrt{M}})$ is a unit vector in $\R^M$ with each coordinate equal to $\frac{1}{\sqrt{M}}$. Then 
	\begin{equation}\label{eqn_prop3.6}
		u(x_0+\delta x_{-1})\ge G(x_0+\delta x_{-1},y_0,H(x_0,y_0,u(x_0))).
	\end{equation}
	For any $x \in \omega+ \delta \overline{B(0,1)}$, from conditions in the proposition, we have $|u(x)|\le R$. Therefore, 
	\begin{flalign*}
	2R &\ge |u(x_0+\delta x_{-1})|+|u(x_0)|&&\\
	&\ge |u(x_0+\delta x_{-1})- u(x_0)|&& \text{(By the triangle inequality)}\\
	&\ge u(x_0+\delta x_{-1}) - u(x_0)&& \\
	&\ge G(x_0+\delta x_{-1}, y_0, H(x_0,y_0,u(x_0)))&& \text{(By inequality \eqref{eqn_prop3.6})  and (by definition }\\
	& -G(x_0, y_0, H(x_0, y_0, u(x_0))) && \text{of $H$, $u(x_0) = G(x_0, y_0, H(x_0, y_0, u(x_0)))$}\\
	&= \int_{0}^{1}\delta \left\langle x_{-1},  D_{x}G\left(x_0+t\delta x_{-1}, y_0, H(x_0,y_0,u(x_0))\right)\right\rangle dt&& \text{(By the fundamental theorem of Calculus)}\\
	&= \frac{\delta}{\sqrt{M}}\int_{0}^{1} \sum\limits_{i=1}^{M}D_{x_i}G\left(x_0+t\delta x_{-1}, y_0, H(x_0, y_0, u(x_0))\right) dt&&\\
	&\ge \frac{\delta}{\sqrt{M}}\int_{0}^{1}\frac{4R\sqrt{M}}{\delta}dt&& \text{(By inequality \eqref{eqn_coercivity})}\\
	&= \frac{\delta}{\sqrt{M}}\cdot\frac{4R\sqrt{M}}{\delta}&&\\
	&= 4R,&&
	\end{flalign*}	
	a contradiction. Therefore, there exists $T>0$ such that for any $x \in \omega$ and $y \in \partial^G u(x)$, one has $\norm{y}\le T$. In addition, here $T = T(\omega, \delta, R)$ is independent of $u$. In fact, from the above argument we can see that $T \le r$, which does not depend on $u$.
\end{proof}

\vspace{0.3cm}

\begin{proof}[Proof of Proposition \ref{proposition:convergence}]
	In this proof, we will show that, under Assumptions  \ref{assmp:Gregular}, \ref{assmp:Gdecreasing}, \ref{assmp:Gcoordinate-monotone}, \ref{assmp:Gtech1}, and \ref{assmp:Gtech3}, the sequence of $G$-convex functions converges, by applying results from Lemma \ref{lemma1}; then we will prove that the limit function is also $G$-convex. 
	Assume $\{u_n\}_{n\in \N}$ is a sequence of $G$-convex functions in $X$ such that for every open convex set $\omega \subset \subset X$, the following holds:
	\begin{equation*}
	\sup\limits_{n} \norm{u_n}_{W^{1,1}(\omega )} < +\infty.
	\end{equation*}

	{\bf Step 1:} By Assumption \ref{assmp:Gtech1}, there exists $k>0$ such that for any $(x, x')\in X^2$, $y\in cl(Y)$, and $z\in cl(Z)$, one has $$\norm{D_xG(x,y,z)-D_x G(x',y,z)} \le k\norm{x-x'}.$$ Denote $G_{\lambda}(x,y,z) := G(x,y,z)+\lambda\norm{x}^2$, where $\lambda \ge \frac{1}{2}\Lip(D_xG)$ with $$\Lip(D_xG)
	:=\sup\limits_{\{(x,x',y,z)\in X\times X\times  cl(Y) \times cl(Z):~x \neq x'\}} \frac{\norm{D_xG(x,y,z)-D_x G(x',y,z)}}{\norm{x-x'}}.$$
	
	Then, for any $(x, x')\in X^2$, by Cauchy–Schwarz inequality, one has 
	\begin{flalign*}
		& \left\langle D_xG_{\lambda}(x,y,z)-D_x G_{\lambda}(x',y,z) , x-x'\right\rangle &&\\
		= & ~\left\langle D_xG(x,y,z)-D_x G(x',y,z) , x-x'\right\rangle + 2\lambda \norm{x-x'}^2 && \text{(By definition of $G_{\lambda}(x,y,z)$)}\\
		\ge & ~-\norm{D_xG(x,y,z)-D_x G(x',y,z)} \norm{x-x'}+ 2\lambda \norm{x-x'}^2 &&\text{(By Cauchy–Schwarz inequality)}\\
		\ge & ~\left[2\lambda - \Lip(D_xG)\right]\norm{x-x'}^2 &&\text{(By definition of $\Lip(D_xG)$)}\\
		\ge & ~0.&&
	\end{flalign*} 
	
	Thus, $G_{\lambda}(\cdot, y, z)$ is a convex function on $X$ for any fixed $(y, z) \in cl(Y) \times cl(Z).$\medskip
	
{\bf Step 2:}	Since $u_n$ is $G$-convex, by Lemma \ref{convex-subdiff}, we know $$u_n(x) = \max\limits_{x'\in X, y\in \partial^G u_n(x')} G(x,y,H(x',y,u_n(x'))).$$ 
Define $v_n(x):= u_n(x) +\lambda \norm{x}^2$. Then 
\begin{flalign*}
	v_n(x) =& \max\limits_{x'\in X, y\in \partial^G u_n(x')}G(x,y,H(x',y,u_n(x'))) +\lambda \norm{x}^2 \\
	=& \max\limits_{x'\in X, y\in \partial^G u_n(x')}\left[G(x,y,H(x',y,u_n(x'))) +\lambda \norm{x}^2\right]\\
	=& \max\limits_{x'\in X, y\in \partial^G u_n(x')} G_{\lambda}(x,y,H(x',y,u_n(x'))).	
\end{flalign*}

	Since $G_{\lambda}(\cdot,y,H(x',y,u_n(x')))$ is convex for each $(x', y)$, we have $v_n(x)$, as supremum of convex functions, is also convex for each $n \in \N$.
	\medskip

{\bf Step 3:}	Since $v_n:= u_n +\lambda\norm{x}^2$ and $\sup\limits_{n}\norm{u_n}_{W^{1,1}(\omega)} < +\infty$, one has 
	$$\sup\limits_{n}\norm{v_n}_{W^{1,1}(\omega)} < +\infty \text{ for any $\omega \subset \subset X$.} $$

	Hence $\{v_n\}_{n\in \N}$ satisfies all the assumptions of Lemma \ref{lemma1}. So, by conclusion of Lemma \ref{lemma1}, there exists a convex function $v^*$ in $X$ and a measurable set $A \subset X$ such that $\dim_{H} (X \setminus A)\le M-1$ and up to a subsequence, $\{v_n\}_{n\in \N}$ converges to $v^*$ uniformly on compact subset of $X$ and $\{D v_n\}_{n\in \N}$ converges to $D v^*$ pointwise in A.
	
	Let $u^*(x):=v^*(x)-\lambda\norm{x}^2$, then  $\{u_n\}_{n\in \N}$ converges to $u^*$ uniformly on compact subset of $X$ and $\{D u_n\}_{n\in \N}$ converges to $D u^*$ pointwise in A.
	\medskip
	
{\bf Step 4:}	Finally, let us prove that $u^*$ is $G$-convex.\medskip

	Define $\Gamma(x):=\cap_{i\ge 1}\overline{\cup_{n\ge i}\partial^G u_n(x)}$ for all $x\in X$.
	\medskip
	
\noindent
{\it Claim}. For any $x'\in X$, we have $\Gamma(x') \neq \emptyset$.
\medskip
	
\begin{proof}
	{\bf Step 4.1:} Let us first show for any $\omega \subset\subset X$, one has $$\sup\limits_{n}\norm{u_n}_{L^{\infty}(\bar{\omega})}<+\infty.$$
	
	Suppose not, then there exits a sequence $\{x_n\}_{n\in \N}\subset \bar{\omega}$ such that $\limsup\limits_{n}|u_n(x_n)|=+\infty$.
	
	Since $\bar{\omega}$ is compact, there exists $\bar{x}\in \bar{\omega}$ such that, up to a subsequence, $x_n\rightarrow \bar{x}$ as $n \rightarrow \infty$. Again up to a subsequence, we may assume that $u_n(x_n)\rightarrow +\infty$ as $n \rightarrow \infty$.
	
	Since $\bar{x} \in \bar{\omega} \subset \subset X$, there exists $\delta >0$ such that $\bar{x}+\delta x_{-1} \in X$, where $x_{-1}:=(\frac{1}{\sqrt{M}}, \frac{1}{\sqrt{M}}, \cdots, \frac{1}{\sqrt{M}})$ is a unit vector in $\R^M$ with each coordinate equal to $\frac{1}{\sqrt{M}}$. For any $x>\bar{x} + \delta x_{-1}$, there exists $n_0\in \N$ such that for any $n>n_0$, we have $x>x_n$. By Proposition \ref{nondecreasing}, $u_n$ are nondecreasing, and thus
	\begin{equation}\label{eqn_integral}
	\int_{\{x\in X|~ x>\bar{x} +\delta x_{-1}\}} u_n(x)dx \ge m\left(\left\{x\in X|~ x> \bar{x}+\delta x_{-1}\right\}\right) u_n(x_n)\rightarrow +\infty.
	\end{equation}
	Here $m\left(\left\{x\in X|~ x>\bar{x}+\delta x_{-1}\right\}\right)$ denotes Lebesgue measure of the set $\{x\in X|~ x>\bar{x}+\delta x_{-1}\}$, which is positive.
	
	Denote $\omega' := \left\{x\in X|~ x>\bar{x}+\delta x_{-1}\right\}$. Therefore, we have $$\norm{u_n}_{W^{1,1}(\omega')} \ge \norm{u_n}_{L^{1}(\omega')} \ge \int_{\omega'} u_n(x) dx \rightarrow +\infty.$$ 
	This implies 
	$$\sup\limits_{n} \norm{u_n}_{W^{1,1}(\omega')} = +\infty.$$
	
	On the other hand, since both $X$ and the set $\left\{x\in \R^M|~ x> \bar{x}+\delta x_{-1} \right\}$ are open and convex, we have  $\omega' = X \cap \{x\in \R^M|~ x> \bar{x}+\delta x_{-1} \}$ is also open and convex. Therefore, by assumption, we have 
	$$\sup\limits_{n} \norm{u_n}_{W^{1,1}(\omega')} < +\infty.$$ 
	
	This is a contradiction. Thus for any $\omega \subset \subset X$, we have $\sup\limits_{n}\norm{u_n}_{L^{\infty}(\bar{\omega})}<+\infty$.\medskip
	
	{\bf Step 4.2:} For any fixed $x'\in X$, there exists an open set $\omega \subset \subset X$ and $\delta>0$ such that $x'\in \omega$ and $\omega + \delta \overline{B(0,1)} \subset \subset X$.
	
	From Step 4.1, we know  $\sup\limits_{n}\norm{u_n}_{L^{\infty}\left(\omega + \delta \overline{B(0,1)}\right)} < +\infty$. 
	There exists $R>0$ such that for all $n\in \N$, we have 
	$$|u_n(x)|\le R \text{ for all $x \in \omega + \delta \overline{B(0,1)}$.}$$ 
	
	Since $u_n$ are $G$-convex functions, by Proposition~\ref{Subdiff/Bdd}, there exists $T = T(\omega, \delta, R) >0$, independent of $n$, such that $\norm{y}\le T$ for any $y \in \partial^G u_n(x')$ and any $n\in \N$. Thus, there exists a sequence $\{y_n\}_{n\in \N}$ such that $y_n \in \partial^G u_n(x')$ and $\norm{y_n}\le T$ for all $n\in \N$.\medskip
	
	By compactness theorem for sequence $\{y_n\}_{n\in \N}$,  there exists $y'$ such that, up to a subsequence, $y_n \rightarrow y'$ as $n \rightarrow \infty$. Thus, we have 
	$$y' \in \overline{\cup_{n\ge i}\partial^G u_n(x')} \text{ for all $i\in \N$. }$$
	It implies $$y' \in \cap_{i\ge 1} \overline{\cup_{n\ge i}\partial^G u_n(x')} = \Gamma (x').$$ 
	
	Therefore $\Gamma(x') \neq \emptyset$ for all $x' \in X$.\medskip
\end{proof}	
	 
	Now for any fixed $x\in X$ and any $y\in \Gamma(x)$, by Cantor's diagonal argument, there exists $\{y_{n_l}\}_{l\in \N}$ such that $$y_{n_l} \in \partial^G u_{n_l}(x) \text{ and } \lim\limits_{l\rightarrow \infty} y_{n_l} = y.$$
	For any $l\in \N$, by definition of $G$-subdifferentiability, one has
	$$u_{n_l}(x')\ge G(x', y_{n_l}, H(x, y_{n_l}, u_{n_l}(x))) \text{	for any $x' \in X$.}$$
	Take limit $l \rightarrow \infty$ at both sides, we get 
	$$u^*(x') \ge G(x', y, H(x, y, u^*(x))) \text{	for any $x'\in X$.}$$ 
	Here we use the fact that both functions $G$ and $H$ are continuous by Assumption \ref{assmp:Gregular} and Proposition \ref{lemma_continuity}. Then by definition of $G$-subdifferentiability, the above inequality implies $y\in \partial ^G u^*(x)$. 

	So $\partial^G u^*(x)\neq \emptyset$ for all $x\in X$, which means $u^*$ is G-subdifferentiable everywhere. By Lemma \ref{convex-subdiff}, $u^*$ is $G$-convex.
\end{proof}

\vspace{0.3cm}

\begin{proof}[Proof of the existence theorem]
{\bf Step 1:} For any $G$-convex $u$, define a point-to-set mapping on $X$ as 
	$$\Phi_u: x \longmapsto argmin_{\partial^G u(x)} \left\{-\pi(x, \cdot, H(x,\cdot,u(x)))\right\}.$$ 
	For each $x \in X$, since $u$ is $G$-convex, $\partial^G u(x) \neq \emptyset$. By definition,  $\partial^G u(x)$ is a closed set. Moreover, by Proposition \ref{Subdiff/Bdd} it is compact. Since both $-\pi$ and $H$ are continuous, $-\pi(x, \cdot, H(x,\cdot,u(x)))$ is also continuous and has minimum on $\partial^G u(x)$. That is, $\Phi_u(x) \neq \emptyset$. Again, since $-\pi(x, \cdot, H(x,\cdot,u(x)))$ is continuous and $\partial^G u(x)$ is compact, $\Phi_u(x)$ is also a compact set. \medskip

	By Proposition \ref{Subdiff/Bdd}, for any compact set $\omega \subset X$, $\cup_{x \in \omega} \partial^G u(x)$ is bounded. Besides, it is compact. $\Phi_u(x)$ is nonempty and compact for all $x\in \omega$, and $\cup_{x \in \omega} \left\{(x, y)|~ y \in \Phi_u(x)\right\}$ is a Borel set. By the measurable selection theorem (cf. \cite[Theorem 1.2, Chapter VIII]{EkelandTemam76}), there exists a measurable mapping $y: \omega \rightarrow Y$ such that for almost all $x$, $y(x) \in \Phi_u(x)$. Let $\{\omega_n\}_{n\in \N}$ denote a sequence of compact sets such that $\omega_1 \subset \omega_2 \subset ... \subset \omega_n \subset ...\subset X$ with $\cup_{n}\omega_n = X$. On each $\omega_n$, there exists a measurable selection map $y^n: \omega_n \rightarrow Y$. Define $\bar{y}: X \rightarrow Y$ such that $\bar{y} = y^1$ on $\omega_1$ and $\bar{y} = y^n$ on $\omega_n\setminus \omega_{n-1}$ for $n \ge 2$. Then $\bar{y}$ is a measurable selection of $\Phi_u$, i.e., $\bar{y}$ is measurable and $\bar{y}(x)\in \Phi_u(x)$ for almost every $x$. 
	\medskip

	Let $\{(u_n, y_n)\}_{n\in \N}$ be a maximizing sequence of $(P)$, where maps $u_n: X\rightarrow \R$ and $y_n: X\rightarrow cl(Y)$ for all $n\in \N$. Without loss of generality, we may assume that for all $n\in \N$, $y_n(\cdot)$ is measurable and $y_n(x) \in \Phi_{u_n}(x)$ for each $x\in X$. Starting from $\{(u_n, y_n)\}_{n\in \N}$, we would find a value-product pair $(u^*, y^*)$ satisfying all the constraints in \eqref{Principal_new_problem}, and show that it is actually a maximizer.
	\medskip
	
{\bf Step 2:} From Assumption \ref{assmp:Gtech0}, there exist $\alpha \ge 1$, $a_1, a_2> 0$, and $b\in \R$ such that for each $x\in X$ and $n \in \N$,
\begin{flalign*}
	a_1 \norm{y_n(x)}_{\alpha}^{\alpha} \le & -\pi(x,y_n(x),H(x, y_n(x), u_n(x))) - a_2 u_n(x) +b \\
	\le &  -\pi(x,y_n(x),H(x, y_n(x), u_n(x)))- a_2 u_{\emptyset}(x) + b.
\end{flalign*}
	Here the second inequality comes from $u_n\ge u_{\emptyset}$. Together with Assumption \ref{assmp:Pi1}, this implies $\{y_n\}_{n\in \N}$ is bounded in $L^{\alpha}(X)$.
	\medskip

	By the participation constraint and Assumption \ref{assmp:Gtech0}, we know 
	\begin{equation*}
	u_{\emptyset}(x) \le u_n(x) = G(x,y_n(x),H(x,y_n(x),u_n(x))) \le \frac{1}{a_2}(b - \pi(x,y_n(x),H(x,y_n(x),u_n(x)))).
	\end{equation*}

	Together with Assumption \ref{assmp:Pi1}, we know $\{u_n\}_{n\in \N}$ is bounded in $L^1(X)$.
	\medskip

	By $G$-subdifferentiability, $Du_n(x) = D_x G(x, y_n(x), H(x,y_n(x),u_n(x)))$. By Assumption \ref{assmp:Gtech2}, we have $$\norm{Du_n}_{1}\le c\norm{y_n}_{\beta}^{\beta}+d \le c\left(N+\norm{y_n}_{\alpha}^{\alpha}\right)+d.$$ The last inequality holds because $\beta \in (0, \alpha]$. Because $X$ is bounded and $\{y_n\}_{n\in \N}$ is bounded in $L^{\alpha}(X)$, we know $\{Du_n\}_{n\in \N}$ is bounded in $L^1(X)$.
	\medskip
	
	Since both $\{u_n\}_{n\in \N}$ and $\{Du_n\}_{n\in \N}$ are bounded in $L^1(X)$, one has $\{u_n\}_{n\in \N}$ is bounded in $W^{1,1}(X)$. By Proposition \ref{proposition:convergence}, there exists a $G$-convex function $u^*$ on $X$ such that, up to a subsequence, $\{u_n\}_{n\in \N}$ converges to $u^*$ in $L^1$ and uniformly on compact subset of $X$, and $\{D u_n\}_{n\in \N}$ converges to $D u^*$ almost everywhere.
	\medskip

{\bf Step 3: } Denote $y^*(x)$ as a measurable selection of $\Phi_{u^*}$. Let us show $(u^*,y^*)$ is a maximizer of the principal's program $(P)$. \medskip

	{\bf Step 3.1: }By Assumption \ref{assmp:Gtech0}, for all $x$, $y_n(x)$ and $u_n(x)$, one has
	\begin{flalign*}
	&-\pi(x,y_n(x),H(x,y_n(x), u_n(x)))\\
	\ge & \ a_2 G(x,y_n(x),H(x,y_n(x), u_n(x))) -b \\
	=&\  a_2 u_n(x) - b \\
	\ge&\ a_2 u_{\emptyset}(x) - b.
	\end{flalign*}

	By Assumption \ref{assmp:Pi1}, $u_{\emptyset}$ is measurable, thus one can apply Fatou's lemma and get
	\begin{align}\label{3}
		\begin{split}
			\sup \tilde{\Pi}(u,y) & = \limsup\limits_{n} \tilde{\Pi}(u_n, y_n) \\
			&= -\liminf\limits_{n} \int_{X} - \pi(x, y_n(x), H(x,y_n(x),u_n(x)))  ~dx\\
			& \le - \int_{X} \liminf\limits_{n} - \pi(x, y_n(x), H(x,y_n(x),u_n(x)))~ dx. \\
		\end{split}
	\end{align}
	
	Define 
	$$\gamma(x):=\liminf\limits_{n} - \pi(x, y_n(x), H(x,y_n(x),u_n(x))). $$
	For each $x\in X$, by extracting a subsequence of $\{y_n\}_{n\in \N}$, which is denoted as $\{y_{n^x_l}\}_{l\in \N}$, we assume 
	$$\gamma(x) = \lim\limits_{l \rightarrow \infty} - \pi(x, y_{n_l^x}(x), H(x,y_{n_l^x}(x),u_{n_l^x}(x))).$$ 
	\medskip
	
	{\bf Step 3.2: } 	For any fixed $x \in X$, since $u_{n_l^x}$ are $G$-convex functions and $\{u_{n_l^x}\}_{l\in \N}$ is bounded in $L^1(X)$, by similar arguments as in the Step 4.1 of the proof of Proposition \ref{proposition:convergence}, $\{u_{n_l^x}\}_{l\in \N}$ is also bounded in $L_{loc}^{\infty}(X)$. Then by Proposition \ref{Subdiff/Bdd}, $\{y_{n_l^x}\}_{l\in \N}$ is also bounded in $L_{loc}^{\infty}(X)$. Thus there exists a subsequence of $\{y_{n_l^x}(x)\}_{l\in \N}$, again denoted as $\{y_{n_l^x}(x)\}_{l\in \N}$, that converges. Denote $\tilde{y}$ a mapping on $X$ such that $y_{n_l^x}(x) \rightarrow \tilde{y}(x)$ as $l \rightarrow \infty$ for all $x\in X$.
	\medskip
	
	Since $\pi$ and $H$ are continuous, we have 
	$$ \gamma(x)= - \pi(x, \tilde{y}(x), H(x,\tilde{y}(x),u^*(x))).$$
	
	For each fixed $x\in X$ and any $l \in \N$, since $u_{n_l^x}$ are $G$-convex and $y_{n_l^x}(x) \in \partial^G u_{n_l^x}(x)$, we have 
	$$u_{n_l^x}(x')\ge G(x', y_{n_l^x}(x),H(x,y_{n_l^x}(x),u_{n_l^x}(x))) \text{  for any $x' \in X$}.$$ 
	Take limit $l \rightarrow \infty$ at both sides, we get 
	$$u^*(x')\ge G(x', \tilde{y}(x),H(x,\tilde{y}(x),u^*(x))) \text{ for any $x'\in X$.}$$
	 By definition of $G$-subdifferentiability, we have $\tilde{y}(x)\in \partial^Gu^*(x)$. 
	 \medskip

{\bf Step 3.3: } By definition of $y^*$, one has 
$$ -\pi(x, y^*(x), H(x,y^*(x),u^*(x)))\le   -\pi(x, \tilde{y}(x), H(x,\tilde{y}(x),u^*(x))) = \gamma(x).$$
	
	So, together with \eqref{3}, we know 
\begin{equation}\label{minimizer}
	\sup \tilde{\Pi}(u,y) \le - \int_{X}  \gamma(x) ~dx \le - \int_{X}  - \pi(x, y^*(x), H(x,y^*(x),u^*(x))) ~dx = \tilde{\Pi}(u^*,y^*).
\end{equation}

	Since $\{u_n\}_{n\in \N}$ converges to $u^*$, and $u_n(x)\ge u_{\emptyset}(x)$ for all $n\in \N$ and $x \in X$, we have $u^*(x)\ge u_{\emptyset}(x)$ for all $x \in X$. In addition, because $u^*$ is $G$-convex and $y^*(x) \in \partial^G u^*(x)$, we know $(u^*, y^*)$ satisfies all the constraints in \eqref{Principal_new_problem}. Together with \eqref{minimizer}, we proved $(u^*,y^*)$ is a solution of the principal's program.
\end{proof}

\bigskip


\bigskip






\begin{thebibliography}{BCDE}
	
	\bibitem{Armstrong96} 
	M. Armstrong, 
	Multiproduct nonlinear pricing, 
	{\em Econometrica}, {\bf 64 }(1996) 51–75.

	\bibitem{Balder77} 
	E.J. Balder, 
	An extension of duality-stability relations to non-convex optimization problems, 
	{\em SIAM J. Control Optim.}, {\bf 15} (1977) 329-343.
	
	\bibitem{BaronMyerson82} 
	D.P. Baron, R.B. Myerson, 
	Regulating a monopolist with unknown costs, 
	{\em Econometrica} {\bf 50} (1982) 911–930.
	
	\bibitem{Basov05} 
	S. Basov, 
	{\em Multidimensional screening}, 
	Springer-Verlag, Berlin, 2005.
	
	\bibitem{Carlier01} 
	G. Carlier, 
	A general existence result for the principal–agent problem with adverse selection, 
	{\em J. Math. Econom.} {\bf 35} (2001) 129–150.
	
	\bibitem{CarlierLachand-Robert01} 
	G. Carlier, T. Lachand-Robert, 
	Regularity of solutions for some variational problems subject to convexity constraint, 
	{\em Comm. Pure Appl. Math.} {\bf 54} (2001) 583–594.
	
	\bibitem{DoleckiKurcyusz78} 
	S. Dolecki, S. Kurcyusz, 
	On $\Phi$-convexity in extremal problems, 
	{\em SIAM J. Control Optim.} {\bf 16} (1978)  277-300.


	\bibitem{EkelandTemam76}
	I. Ekeland, R. Temam,
	 {\em Analyse convexe et probl\'emes variationnels},
	 Dunod (Libraire), Paris, 1976.


	\bibitem{ElsterNehse74} 
	K.-H. Elster, R. Nehse, 
	Zur theorie der polarfunktionale, 
	{\em Math. Operationsforsch. Stat.} {\bf 5} (1974) 3-21.


	\bibitem{Evans98}
	L. C. Evans, 
	{\em Partial differential equations},
	American Mathematical Society, Providence, Rhode Island, 1998.
	

	\bibitem{FigalliKimMcCann11} 
	A. Figalli, Y.-H. Kim, R.J. McCann, 
	When is multidimensional screening a convex program? 
	{\em J. Econom. Theory} {\bf 146} (2011) 454-478.
	
	\bibitem{GangboMcCann96} 
	W. Gangbo, R.J. McCann, 
	The geometry of optimal transportation, 
	{\em Acta Math.} {\bf 177} (1996) 113–161.

	\bibitem{KutateladzeRubinov72} 
	S.S. Kutateladze, A.M. Rubinov, 
	Minkowski duality and its applications, 
	{\em Russian Math. Surveys} {\bf 27}  (1972) 137-192.
	
	\bibitem{MartinezLegaz05} 
	J.E. Martínez-Legaz, 
	Generalized convex duality and its economic applications, in: 
	{\em Handbook of generalized convexity and generalized monotonicity}, Springer, New York, 2005, pp. 237--292.	

	\bibitem{MaskinRiley84} 
	E. Maskin, J. Riley, 
	Monopoly with incomplete information, 
	{\em The RAND Journal of Economics} {\bf 15} (1984) 171-196.

	\bibitem{McAfeeMcMillan88} 
	R.P. McAfee, J. McMillan, 
	Multidimensional incentive compatibility and mechanism design, 
	{\em J. Econom. Theory} {\bf 46} (1988) 335–354.
		
	\bibitem{McCannZhang17}
	R.J. McCann, K.S. Zhang,
	On concavity of the monopolist's problem facing consumers with nonlinear price preferences,
	To appear in {\em Comm. Pure and Applied Math.}
		
	\bibitem{Mirrlees71}  
	J.A. Mirrlees, 
	An exploration in the theory of optimum income taxation, 
	{\em Rev. Econom. Stud.} {\bf 38} (1971) 175–208.

	\bibitem{MonteiroPage98}  
	P.K. Monteiro, F.H. Page Jr., 
	Optimal selling mechanisms for multiproduct monopolists: incentive compatibility in the presence of budget constraints, 
	{\em J. Math. Econom.} {\bf 30} (1998) 473–502.


	\bibitem{MussaRosen78} 
	M. Mussa, S. Rosen, 
	Monopoly product and quality, 
	{\em J. Econom. Theory} {\bf 18} (1978) 301–317.

	\bibitem{Myerson81}
	R.B. Myerson, 
	Optimal auction design, 
	{\em  Mathematics of Operations Research} {\bf 6} (1981) 58-73.

	\bibitem{NoldekeSamuelson15p} 
	G. N\"oldeke, L. Samuelson, 
	The implementation duality, 
	{\em Econometrica} {\bf 86(4)} (2018) 1283–1324.

	\bibitem{Rochet87}
	J.-C. Rochet,
	A necessary and sufficient condition for rationalizability in a quasi-linear context,
	{\em J. Math. Econom.} {\bf 16} (1987) 191--200.

	\bibitem{RochetChone98} 
	J.-C. Rochet, P. Chon\'e, 
	Ironing sweeping and multidimensional screening, 
	{\em Econometrica} {\bf 66} (1998) 783–826.
	
	\bibitem{RochetStole03} 
	J.-C. Rochet, L.A. Stole, 
	The economics of multidimensional screening, in: 
	{\em M. Dewatripont, L.P. Hansen, S.J. Turnovsky (Eds.),  Advances in Economics and Econometrics}, Cambridge University Press, Cambridge, 2003, pp. 150-197.
	
	\bibitem{Rubinov00a} 
	A.M. Rubinov, 
	Abstract convexity: Examples and applications, 
	{\em Optimization} {\bf 47} (2000)  1–33. 


	\bibitem{Singer97} 
	I. Singer, 
	{\em Abstract convex analysis}, Wiley-Interscience, New York, 1997.

	\bibitem{Spence74} 
	M. Spence, 
	Competitive and optimal responses to signals: An analysis of efficiency and distribution, 
	{\em J. Econom. Theory} {\bf 7} (1974) 296–332.

	\bibitem{Spence80} 
	M. Spence, 
	Multi-product quantity-dependent prices and profitability constraints, 
	{\em Rev. Econom. Stud.} {\bf 47} (1980) 821–841.
	
	\bibitem{Trudinger14} 
	N. S. Trudinger, 
	On the local theory of prescribed Jacobian equations, 
	{\em Discrete Contin. Dyn. Syst.} {\bf 34} (2014) 1663-1681.

	\bibitem{Wilson93} 
	R. Wilson, 
	{\em Nonlinear pricing}, 
	Oxford University Press, Oxford, 1993.

	\bibitem{Zhang18}
	K.S. Zhang,
	{\em Existence, uniqueness, concavity and geometry of the monopolist’s
	problem facing consumers with nonlinear price preferences}, Ph.D. Thesis, University of Toronto, 2018.
	
	
\end{thebibliography}
\end{document}